\documentclass[12pt, draftclsnofoot, onecolumn]{IEEEtran}
\usepackage{cite}

\let\hbarorig\hbar % Alan, 20150824
\usepackage{amsmath}
\usepackage{amsthm}
\usepackage{amsfonts}
\usepackage{amssymb}
\let\hbar\hbarorig % Alan, 20150824
\usepackage{balance}
\usepackage[singlelinecheck=on]{caption}
\usepackage{verbatim}

\usepackage[lined,boxed,commentsnumbered]{algorithm2e}
\usepackage{setspace}

\usepackage{multirow}
\usepackage{framed}
\usepackage{enumitem}

\usepackage{graphicx}
\usepackage{epstopdf}
\graphicspath{{fig/}}
\usepackage{calc}
\usepackage{color}
\usepackage[caption=false,font=footnotesize]{subfig}

\newtheorem{theorem}{Theorem}
\newtheorem{lemma}{Lemma}
\newtheorem{remark}{Remark}

\newtheorem{example}{Example}
\newtheorem{definition}{Definition}

\newtheorem{stopping criterion}{stopping criterion}

\newcommand{\beq}{\begin{equation}}
\newcommand{\eeq}{\end{equation}}
\newcommand{\beqnn}{\begin{equation*}}
\newcommand{\eeqnn}{\end{equation*}}
\newcommand{\beqy}{\begin{eqnarray}}
\newcommand{\eeqy}{\end{eqnarray}}
\newcommand{\beqynn}{\begin{eqnarray*}}
\newcommand{\eeqynn}{\end{eqnarray*}}
\newcommand{\bit}{\begin{itemize}}
\newcommand{\eit}{\end{itemize}}
\newcommand{\ben}{\begin{enumerate}}
\newcommand{\een}{\end{enumerate}}
\newcommand{\bex}{\begin{example}}
\newcommand{\eex}{\end{example}}

\newcommand{\balg}[1]{\begin{algorithm} \caption{#1}}
\newcommand{\ealg}{\end{algorithm}}

\newcommand{\balgc}{\begin{algorithmic}[1]}
\newcommand{\ealgc}{\end{algorithmic}}

\newcommand{\bary}{\begin{array}}
\newcommand{\eary}{\end{array}}
\newcommand{\bmx}{\begin{bmatrix}}
\newcommand{\emx}{\end{bmatrix}}
\newcommand{\bsmx}{\left[\begin{smallmatrix}}
\newcommand{\esmx}{\end{smallmatrix}\right]}
\newcommand{\bmxc}[1]{\left[\begin{array}{@{}#1@{}}}
\newcommand{\emxc}{\end{array}\right]}
\newcommand{\bcn}{\begin{center}}
\newcommand{\ecn}{\end{center}}

\newcommand{\Rbb}{{\mathbb{R}}}
\newcommand{\Zbb}{{\mathbb{Z}}}
\newcommand{\Cbb}{{\mathbb{C}}}

\newcommand{\bigN}{{\mathcal{N}}}
\newcommand{\bigR}{{\mathcal{R}}}

\newcommand{\bsG}{\boldsymbol{G}}

\newcommand{\I}{\boldsymbol{I}}

\renewcommand{\P}{\boldsymbol{P}}

\renewcommand{\S}{\boldsymbol{S}}
\newcommand{\T}{\boldsymbol{T}}

\newcommand{\W}{\boldsymbol{W}}

\renewcommand{\a}{\boldsymbol{a}}

\newcommand{\e}{\boldsymbol{e}}

\newcommand{\h}{\boldsymbol{h}}

\newcommand{\p}{\boldsymbol{p}}

\newcommand{\rr}{\boldsymbol{r}}
\newcommand{\bsr}{\boldsymbol{r}} % Alan ZHOU, 20141124

\newcommand{\bst}{\boldsymbol{t}} % Alan ZHOU, 20140812
\renewcommand{\u}{\boldsymbol{u}}

\newcommand{\w}{\boldsymbol{w}}
\newcommand{\x}{{\boldsymbol{x}}}
\newcommand{\y}{{\boldsymbol{y}}}
\newcommand{\z}{\boldsymbol{z}}
\newcommand{\0}{{\boldsymbol{0}}}

\newcommand{\norm}[1]{\left\lVert #1 \right\rVert} % norm: ||*||
\newcommand{\floor}[1]{\left\lfloor #1 \right\rfloor} % floor
\newcommand{\ceil}[1]{\left\lceil #1 \right\rceil} % ceil
\newcommand{\round}[1]{\left\lfloor #1 \right\rceil} % round
\newcommand{\fabs}[1]{{\rm abs}\left( #1 \right)} % abs: abs(*)
\newcommand{\fsign}[1]{{\rm sign}\left( #1 \right)} % sign: sign(*)

\newcommand{\aStar}{\a^\star} % Make sure \a is already defined, say, in micros.tex.
\newcommand{\aDiamond}{\a^\diamond}
\newcommand{\aBar}{\bar{\a}}
\newcommand{\aBarStar}{\aBar^\star}
\newcommand{\aBarDagger}{\aBar^\dagger}
\newcommand{\aBarDiamond}{\aBar^\diamond}
\newcommand{\hBar}{\bar{\h}}

\begin{document}

\title{A Quadratic Programming Relaxation Approach to Compute-and-Forward\\Network Coding Design}

\author{Baojian~Zhou,~\IEEEmembership{Student~Member,~IEEE,}
Jinming~Wen,
and~Wai~Ho~Mow,~\IEEEmembership{Senior~Member,~IEEE}% <-this % stops a space
\thanks{Baojian~Zhou and Wai~Ho~Mow are with the Department
of Electronic and Computer Engineering, The Hong Kong University of Science and Technology, Clear Water Bay,
NT, Hong Kong (e-mail: \{bzhouab, eewhmow\}@ust.hk).
Baojian~Zhou was supported by a grant from University Grants Committee
of the Hong Kong Special Administrative Region, China (Project No. AoE/E-02/08).}%
\thanks{Jinming~Wen is with the Laboratoire de l'Informatique du Parall\'elisme, (CNRS,
ENS de Lyon, Inria, UCBL), Universit\'e de Lyon, Lyon 69007, France (e-mail: jwen@math.mcgill.ca).
Jinming~Wen was supported in part by ANR through the HPAC project
under Grant ANR 11 BS02 013.}}% <-this % stops a space}
\maketitle

\begin{abstract}
Using physical layer network coding, compute-and-forward is a promising relaying scheme
that effectively exploits the interference between users and thus achieves high rates.
In this paper, we consider the problem of finding the optimal integer-valued coefficient vector for a relay
in the compute-and-forward scheme to maximize the computation rate at that relay.
Although this problem turns out to be a shortest vector problem, which is suspected to be NP-hard,
we show that it can be relaxed to a series of equality-constrained quadratic programmings.
The solutions of the relaxed problems serve as real-valued approximations of the optimal coefficient vector,
and are quantized to a set of integer-valued vectors, from which a coefficient vector is selected.
The key to the efficiency of our method is that the closed-form expressions of the real-valued approximations
can be derived with the Lagrange multiplier method.
Numerical results demonstrate that compared with the existing methods,
our method offers comparable rates at an impressively low complexity.
\end{abstract}


\IEEEpeerreviewmaketitle

\section{Introduction}
\IEEEPARstart{B}{eing} a promising relaying strategy in wireless networks, compute-and-forward (CF)
has attracted a lot of research interest since it was proposed in 2008 by Nazer and Gastpar \cite{Nazer2008}.
The advantage of CF is that it achieves higher rates in the medium signal-to-noise ratio (SNR) regime
when compared with other relaying strategies, e.g., amplify-and-forward, decode-and-forward.
Relays in CF attempt to decode
integer linear combinations of the transmitted codewords, rather than the codewords themselves.
The integer coefficient vectors corresponding to the linear combinations and the decoded messages are
then forwarded to the destination.
Under certain conditions (see \cite{Nazer2011} for details), the destination can recover the original source messages
with enough forwarded messages and coefficient vectors from the relay.

The design of the CF scheme lies in selecting the coefficient vectors at the relays.
There are many choices of coefficient vectors for one relay, and each may render
a different \emph{computation rate} \cite{Nazer2011} at that relay.
Computation rate is defined as the maximum transmission rate from the associated sources to a relay
such that the linear combinations at the relay can be successfully decoded.
If the coefficient vectors are linearly independent,
the \emph{achievable rate} of the network equals the minimum computation rate;
otherwise, the achievable rate of the network is zero and none of the messages can be recovered.
The objective of designing the CF scheme is
to maximize the overall achievable rate of the network.
One approximation method is presented in \cite{Wei2012}.
However, for those networks
where each relay is allowed to send only one coefficient vector to the destination,
and only local channel state information (CSI) is available,
i.e., each relay knows merely its own channel vector,
one reasonable solution is to select the coefficient vector
that maximizes the computation rate at each relay.

In this paper, we consider additive white Gaussian noise (AWGN) networks where only local CSI is available,
and focus on the CF network coding design problem
with the objective being maximizing the computation rate at a relay by choosing the optimal coefficient vector.
It has been shown that this problem reduces to a shortest vector problem (SVP).
Different methods have been developed to tackle this problem.
The branch-and-bound method proposed in \cite{Richter2012} finds the optimal solution
but its efficiency degrades as the dimension of channel vectors grows according to the simulation results.
Although the general SVP is suspected to be NP-hard,
Sahraei and Gastpar showed in \cite{Sahraei2014} that the SVP in the CF design is special,
and developed an algorithm (called the ``SG'' method in this paper) that solves the SVP in polynomial time.
For independent and identically distributed (i.i.d.) Gaussian channel entries,
the complexity of the SG method is of order 2.5 with respect to the dimension,
and is linear with respect to the square root of the signal-to-noise ratio (SNR).
A class of methods are those based on lattice reduction (LR) algorithms
(e.g., Minkowski, HKZ, LLL, and CLLL LR algorithms;
c.f. \cite{Zhang2012, Lenstra1982, Vetter2009, Ling2013, Chang2013, Gan2009}).
The method in \cite{Sakzad2012} based on the LLL LR algorithm \cite{Lenstra1982} provides close-to-optimal rates
and is well-known to be of polynomial time complexity with respect to the vector dimension.
In \cite{Wen2015}, we proposed an efficient method based on sphere decoding to find the optimal coefficient vector;
however there is no theoretical guarantee on the complexity.

Our goal in this work is to develop a new method that finds a suboptimal coefficient vector for a relay
with low complexity compared with the existing methods,
while provides a close-to-optimal computation rate at the same time.
Taking advantage of some useful properties of the problem, we first show that the original SVP can be
approximated by a series of quadratic programmings (QPs).
The closed-form solutions of the QPs are derived by use of the Lagrange multiplier method and
can be computed with linear complexity with respect to the dimension,
which is the key to the efficiency of our method.
The solutions of the QPs serve as real-valued approximations of the integer coefficient vector,
and are quantized into a set of candidate integer vectors by a successive quantization algorithm.
Finally, the integer vector in the candidate set that maximizes the computation rate
is selected to be the coefficient vector.
The complexity of our method is of order 1.5 with respect to the dimension for i.i.d. Gaussian channel entries,
and is lower than the above mentioned methods.
Numerical results demonstrate that among existing methods that provide close-to-optimal rates,
our method is much more efficient as expected.

As a summary, our contributions in this work include the following:
\begin{itemize}
\item For the real-valued channels,
we develop a quadratic programming relaxation approach to find a suboptimal coefficient vector
for a relay so that the computation rate at that relay is close-to-optimal.
The complexity is $O(L\sqrt{P\norm{\h}^2})$ for a given channel vector $\h\in\Rbb^L$
and signal power constraint $P$,
and is of average value $O(P^{0.5}L^{1.5})$ for i.i.d. standard Gaussian channel entries.
\item For the complex-valued channels, we demonstrate how to apply our method in an efficient way
to find the complex-valued coefficient vector.
\item Extensive simulation results are presented to compare the effectiveness and efficiency of our method
with the existing methods.
\end{itemize}
Part of this work has been presented in \cite{Zhou2014}.
One main improvement here is the complexity order for i.i.d. Gaussian channel entries is further reduced from 3 to 1.5.

%As a comparison, new contributions in this work include the following.
%\begin{itemize}
%\item The complexity for $L$-dimensional channel vector with i.i.d. Gaussian entries
%is further reduced from $O(L^3)$ to $O(L^{1.5})$,
%by avoiding direct calculation of the matrix inversion in the closed-form expression of the real-valued approximations,
%as well as by simplifying the inequality condition in the quantization algorithm.
%\item A better criterion for determine the number of real-valued approximations is established,
%which improves the efficiency in low SNR regime.
%\item Application of the proposed method in the complex-valued channel model is now covered,
%and this makes the method be of more practical use.
%\item More extensive numerical results are presented in the sense that
%1) more new methods are included for comparison purposes, and
%2) a wider range of dimensions, from 2 to 16, are considered.
%\end{itemize}

In the following, we will first introduce the system model of AWGN networks as well as the CF network coding
design problem in Section~\ref{section:ProblemStatement}.
Then in Section~\ref{section:ProposedMethod}, we will present our proposed method in detail.
Numerical results will be shown in Section~\ref{section:NumericalResults}.
Finally, we will conclude our work in Section~\ref{section:Conclusions}.

{\it Notation.}
Let $\Rbb$ be the real field, $\Cbb$ be the complex field, and $\Zbb$ be the ring of integers.
Boldface lowercase letters denote column vectors, and boldface uppercase letters denote matrices,
e.g., $\w\in\Rbb^L$ and $\W\in\Rbb^{M\times L}$.
$\norm{\w}$ denotes the $\ell^2$-norm of $\w$, and $\w^T$ denotes the transpose of $\w$.
For a vector $\w$, let $\w(\ell)$ be the element with index $\ell$,
and $\w(i\!:\!j)$ be the vector composed of elements with indices from $i$ to $j$.
For a matrix $\W$, let \mbox{$\W(i\!:\!j,k\!:\!\ell)$} be the submatrix containing elements
with row indices from $i$ to $j$ and column indices from $k$ to $\ell$,
\mbox{$\W(i\!:\!j,k)$} be the submatrix containing elements with row indices from $i$ to $j$ and column index $k$,
$\W(i,k\!:\!\ell)$ be the submatrix containing elements with row index $i$ and column indices from $k$ to $\ell$,
and $\W(i,j)$ be the element with row index $i$ and column index $j$.
Let $\floor{x}$ and $\ceil{x}$, i.e., the corresponding floor and ceiling functions of $x$,
be the maximum integer no greater than $x$ and the minimum integer no less than $x$, respectively.
Let $\floor{\w}_\ell$ and $\ceil{\w}_\ell$ be the vectors generated from $\w$
by applying the corresponding operation on the $\ell$\!-th element only.
$\0$ denotes an all-zero vector, and $\I$ denotes an identity matrix.
${\rm sign}(\w)$ returns the vector that contains the signs of the elements in $\w$.
${\rm abs}(\w)$ returns the vector whose elements are the absolute values of the elements in $\w$.

\section{Problem Statement}
\label{section:ProblemStatement}
We consider additive white Gaussian noise (AWGN) networks \cite{Nazer2011} where sources, relays and destinations
are connected by linear channels with AWGN.
For the ease of explanation, we first develop our method for real-valued channels,
and then demonstrate how to apply our method to complex-valued channels.
An AWGN network with real-valued channels is defined as the following.

\begin{definition}
\label{definition:RealChannelModel}
\emph{(Real-Valued Channel Model)}
In an AWGN network, each relay (indexed by $m=1,2,\cdots,M$) observes a noisy linear combination of
the transmitted signals through the channel,
\begin{align}
\label{equation:RealChannelModel}
\y_m = \sum_{\ell=1}^L \h_m(\ell)\x_\ell + \z_m,
\end{align}
where $\x_\ell \in \Rbb^n$ with the power constraint $\frac{1}{n}\norm{\x_\ell}^2 \leq P$
is the transmitted codeword from source $\ell$ ($\ell = 1,2,\cdots,L$),
$\h_m \in \Rbb^L$ is the channel vector to relay $m$,
$\h_m(\ell)$ is the $\ell$-th entry of $\h_m$,
$\z_m \in \Rbb^n$ is the noise vector with entries being i.i.d. Gaussian,
i.e., $\z_m\!\sim\!\bigN\!\left(\0,\I\right)$, and $\y_m$ is the signal received at relay $m$.
\end{definition}

In the sequel, we will focus on one relay and thus ignore the subscript ``$m$'' in $\h_m$, $\a_m$, etc.

In CF, rather than directly decode the received signal $\y$ as a codeword, a relay first applies to $\y$ an amplifying
factor $\alpha$ such that $\alpha\h$ is close to an integer \emph{coefficient vector} $\a$,
and tries to decode $\alpha\y$ as an integer linear combination,
whose coefficients form $\a$, of the original codewords $\{\x_\ell\}$.
The \emph{computation rate}~\cite{Nazer2011} is the maximum transmission rate from the associated sources to a relay
such that the integer linear combinations at the relay can be decoded with arbitrarily small error probability.
Assume the $\log$ function is with respect to base 2,
and define $\log^+(w) \triangleq \max \left( \log(w),0 \right)$.
The computation rate can be calculated with Theorem~\ref{theorem:ComputationRate} from~\cite{Nazer2011}.

\begin{theorem}
\emph{(Computation Rate in Real-Valued Channel Model)}
\label{theorem:ComputationRate}
For a relay with coefficient vector $\a$ in the real-valued channel model defined in
Definition~\ref{definition:RealChannelModel},
the following computation rate is achievable,

\begin{equation}
\begin{aligned}
\label{equation:ComputationRate}
\bigR \left( \h, \a \right) =
\frac{1}{2} \log^+
\left( \left(
\norm{\a}^2 -
\frac{P (\h^T\a)^2}{1 + P\norm{\h}^2}
\right)^{-1}\right).
\end{aligned}
\end{equation}
\end{theorem}

With the computation rate being the metric, we define the optimal coefficient vector as follows.

\begin{definition}
\emph{(The Optimal Coefficient Vector)}
The optimal coefficient vector $\aStar$ for a channel vector $\h$ is the one that maximizes the computation rate,
\begin{equation}
\begin{aligned}
\label{equation:aOptimal}
\aStar = \arg \max_{\a \in \Zbb^L \backslash\{\0\}} \bigR \left( \h, \a \right).
\end{aligned}
\end{equation}
\end{definition}

After a few simple manipulations,
the optimization problem stated in \eqref{equation:aOptimal} can be written in the following quadratic form~\cite{Wei2012},
\begin{equation}
\label{equation:IntegerQP}
\begin{aligned}
\aStar =
\arg\min_{\a\in\Zbb^L\backslash\{\0\}}
\a^T\bsG\a,
\end{aligned}
\end{equation}
where
\begin{equation}
\begin{aligned}
\label{equation:G}
\bsG \triangleq
\I - \frac{P}{1 + P\norm{\h}^2} \h\h^T.
\end{aligned}
\end{equation}

If we take $\bsG$, which is positive definite, as the \emph{Gram matrix} of a lattice $\Lambda$,
then the problem turns out to be the SVP in the lattice $\Lambda$.
In the next section, we will propose an efficient approximation method based on QP relaxation
that gives a suboptimal coefficient vector.

\section{Proposed Method}
\label{section:ProposedMethod}
In this section, we will first derive our method for the real-valued channel model,
and then extend the method for the complex-valued channel model.

\subsection{Preliminaries}
We start with investigating some properties of the problem, which is the basis of our new method.

\begin{definition}
(Signature Matrix)
A signature matrix is a diagonal matrix whose diagonal elements are $\pm1$.
\end{definition}

\begin{definition}
(Signed Permutation Matrix)
A signed permutation matrix is a generalized permutation matrix whose nonzero entries are $\pm 1$.
\end{definition}

\begin{remark}
After replacing $-1$'s with $1$'s, a signed permutation matrix becomes a permutation matrix.
Obviously, signed permutation matrices are unimodular and orthogonal.
Every signed permutation matrix can be expressed as $\S=\P\T$,
where $\T$ is a signature matrix, and $\P$ is a permutation matrix.
\end{remark}

\begin{theorem}
\label{theorem:ProblemTransformation}
If $\aStar$ is the optimal coefficient vector for a channel vector $\h$ with power constraint $P$,
then for any signed permutation matrix $\S\in\Zbb^{L\times L}$, $\S\aStar$ is optimal for $\S\h$
with the same power constraint $P$, and $\bigR(\h,\aStar)=\bigR(\S\h,\S\aStar)$.
\end{theorem}

\begin{proof}
We first show $\bigR \left( \h, \a \right) = \bigR \left( \S\h, \S\a \right)$ for any $\h$ and $\a$
with the same power constraint $P$.
$\S$ is unimodular, then $\S\a$ is an integer vector and can be applied as a coefficient vector.
$\S$ is orthogonal, then $\S^T\S = \I$. $\norm{\S\h}^2 = \h^T\S^T\S\h = \h^T\h = \norm{\h}^2$,
and similarly $\norm{\S\a}^2 = \norm{\a}^2$. $(\S\h)^T\S\a = \h^T\S^T\S\a = \h^T\a$.
According to Theorem~\ref{theorem:ComputationRate}, the computation rate $\bigR \left( \h, \a \right)$
is determined by $P$, $\norm{\h}^2$, $\norm{\a}^2$, and $\h^T\a$.
Thus, $\bigR \left( \h, \a \right) = \bigR \left( \S\h, \S\a \right)$.

$\aStar$ is optimal for $\h$ means $\aStar$ maximizes $\bigR \left( \h, \a \right)$.
Then $\S\aStar$ maximizes $\bigR \left( \S\h, \S\a \right)$
since $\bigR \left( \h, \a \right) = \bigR \left( \S\h, \S\a \right)$ always holds.
Therefore, $\S\aStar$ is optimal for $\S\h$ with the same power constraint $P$.
\end{proof}

\begin{definition}
(Nonnegative Ordered Vector)
A vector $\h$ is said to be nonnegative ordered if its elements are nonnegative and
in nondecreasing order according to their indices.
\end{definition}

\begin{lemma}
\label{lemma:h2hBar}
For any vector $\h$, there exists a signed permutation matrix $\S$ such that $\S\h$ is nonnegative ordered.
\end{lemma}

\begin{remark}
To find such an $\S$ in Lemma~\ref{lemma:h2hBar}, we can simply choose $\S=\P\T$,
where $\T$ is a signature matrix that converts all the elements in $\h$ to nonnegative,
and $\P$ is a permutation matrix that sorts the elements in $\T\h$ in nondecreasing order.
\end{remark}

With Theorem~\ref{theorem:ProblemTransformation} and Lemma~\ref{lemma:h2hBar},
for any channel vector $\h$, we can first find a signed permuation matrix $\S$
and transform $\h$ to the nonnegative ordered $\hBar=\S\h$,
then obtain the optimal coefficient vector $\aBarStar$ for $\hBar$,
and finally recover the desired optimal coefficient vector $\aStar=\S^{-1}\aBarStar$ for $\h$.
In this way, it suffices to focus on solving the problem in \eqref{equation:IntegerQP}
for nonnegative ordered channel vectors $\hBar$.

\begin{remark}
\label{remark:Transformation}
In implementation, there is no need to use the signed permutation matrix $\S$.
It is merely necessary to:
1) record the sign of the elements in $\h$ with a vector $\bst=\fsign{\h}$,
and 2) sort $\boldsymbol{\hbar}=\fabs{\h}$ in ascending order as $\hBar$
and record the original indices of the elements with a vector $\p$
such that $\hBar(\ell)=\boldsymbol{\hbar}(\p(\ell))$, $\ell=1,2,\cdots,L$.
After $\aBarStar$ for $\hBar$ is obtained, $\aStar$ for $\h$ can be recovered with
$\aStar(\p(\ell))=\bst(\p(\ell))\aBarStar(\ell)$, $\ell=1,2,\cdots,L$.
\end{remark}

\begin{example}
Given a channel vector as $\h=[-1.9,0.1,1.1]^T$, then $\bst=[-1,1,1]^T$, $\fabs{\h}=[1.9,0.1,1.1]^T$,
$\hBar=[0.1,1.1,1.9]^T$, and $\p=[2,3,1]^T$.
If for certain power $P$, $\aBarStar=[0,1,2]^T$, then $\aStar=[-2,0,1]^T$.
\end{example}

According to \eqref{equation:IntegerQP}, if $\aStar$ is optimal for $\h$, then $-\aStar$ is also optimal for $\h$.
To reduce redundancy, we restrict the optimal coefficient vector $\aStar$ to be the one
such that $\h^T\aStar\geq0$ in the following.

\begin{lemma}
\label{lemma:aNonnegative}
If all the elements in a channel vector $\h$ are nonnegative,
then all the elements in the optimal coefficient vector $\aStar$ are also nonnegative.
\end{lemma}
\begin{proof}
Suppose $\aStar(i) < 0$, and define $\a'$ as: $\a'(i) = 0$, and $\a'(\ell) = \aStar(\ell)$, $\forall \ell \neq i$.
Obviously, $\norm{\a'} < \norm{\aStar}$, and $\h^T \a' \geq \h^T \aStar \geq 0$.
Then according to \eqref{equation:ComputationRate}, $\bigR \left( \h, \a' \right) > \bigR \left( \h, \aStar \right)$,
which implies $\aStar$ is not optimal and leads to a contradiction.
Thus, all the elements in $\aStar$ must be nonnegative.
\end{proof}

\begin{lemma}
\label{lemma:a0}
For a channel vector $\h$ and its corresponding optimal coefficient vector $\aStar$,
if $\h(i) = 0$, then $\aStar(i) = 0$.
\end{lemma}
\begin{proof}
Suppose $\h(i) = 0$, and $\aStar(i) \neq 0$.
Define $\a'$ as: $\a'(i) = 0$, and $\a'(\ell) = \aStar(\ell)$, $\forall \ell \neq i$.
Obviously, $\norm{\a'} < \norm{\aStar}$, and $\h^T \a' = \h^T \aStar \geq 0$.
Then according to \eqref{equation:ComputationRate}, $\bigR \left( \h, \a' \right) > \bigR \left( \h, \aStar \right)$,
which implies $\aStar$ is not optimal.
Thus, if $\h(i) = 0$, then $\aStar(i) = 0$.
\end{proof}

\begin{lemma}
\label{lemma:hijEqual}
For a channel vector $\h$ and its corresponding optimal coefficient vector $\aStar$,
if $\h(i) = \h(j)$, $i < j$, then $\aStar(i) = \aStar(j)$ or $\fabs{\aStar(i) - \aStar(j)} = 1$.
\end{lemma}
\begin{proof}
Without loss of generality, assume $\aStar(i) - \aStar(j) < -1$.
Define $\a'$ as: $\a'(i) = \aStar(i)+1$, $\a'(j) = \aStar(j)-1$, and $\a'(\ell) = \aStar(\ell)$, $\forall \ell \notin \{i,j\}$.
Obviously, $\norm{\a'} < \norm{\aStar}$, and $\h^T \a' = \h^T \aStar \geq 0$.
Then according to \eqref{equation:ComputationRate}, $\bigR \left( \h, \a' \right) > \bigR \left( \h, \aStar \right)$,
which implies $\aStar$ is not optimal.
Thus, $\aStar(i) - \aStar(j) \geq -1$. Similarly, $\aStar(j) - \aStar(i) \geq -1$.
Therefore, $\aStar(i) = \aStar(j)$ or $\fabs{\aStar(i) - \aStar(j)}=~1$.
\end{proof}

\begin{remark}
\label{remark:hijEqual}
In Lemma~\ref{lemma:hijEqual}, for the case where $\h(i) = \h(j)$ with $\fabs{\aStar(i) - \aStar(j)} = 1$, $i < j$,
we will always set $\aStar(j) = \aStar(i) + 1$ since setting $\aStar(i) = \aStar(j) + 1$ results the same computation rate.
Then, as long as $\h(i) = \h(j)$, $i < j$, it holds that $\aStar(i)\leq\aStar(j)$.
\end{remark}

\begin{theorem}
\label{theorem:aNonnegativeOrdered}
For a nonnegative ordered channel vector $\h$, the optimal coefficient vector $\aStar$ is also nonnegative ordered.
\end{theorem}
\begin{proof}
According to Lemma~\ref{lemma:aNonnegative}, all the elements in $\aStar$ are nonnegative.
Suppose $\aStar$ is not nonnegative ordered,
then there must exist $i, j $ ($1 \leq i < j \leq L $) such that $\aStar(i) > \aStar(j) \geq 0$.
According to Lemma~\ref{lemma:a0}, $\aStar(i) > 0$ implies $\h(i) > 0$.
According to Lemma~\ref{lemma:hijEqual} and Remark~\ref{remark:hijEqual},
$\aStar(i) > \aStar(j)$ implies $\h(i) \neq \h(j)$ and thus $\h(j) > \h(i) > 0$.
Then, $\h(i) \aStar(j) + \h(j) \aStar(i) > \h(i) \aStar(i) + \h(j) \aStar(j)$.

Define $\a'$ as: $\a'(i) = \aStar(j)$, $\a'(j) = \aStar(i)$, and $\a'(\ell) = \aStar(\ell)$, $\forall \ell \notin \{i,j\}$.
Obviously, $\norm{\a'} = \norm{\aStar}$, and
$\h^T \a'
= \sum_{\ell = 1}^L \h(\ell) \a'(\ell)
= \sum_{\ell \notin \{i,j\}} \h(\ell) \a'(\ell) + \h(i) \a'(i) + \h(j) \a'(j)
= \sum_{\ell \notin \{i,j\}} \h(\ell) \aStar(\ell) + \h(i) \aStar(j) + \h(j) \aStar(i)
> \sum_{\ell \notin \{i,j\}} \h(\ell) \aStar(\ell) + \h(i) \aStar(i) + \h(j) \aStar(j)
= \sum_{\ell = 1}^L \h(\ell) \aStar(\ell)
= \h^T \aStar \geq 0$.
Then according to \eqref{equation:ComputationRate}, $\mathcal{R} \left( \h, \a' \right) > \mathcal{R} \left( \h, \aStar \right)$,
which implies $\aStar$ is not optimal.
Therefore, $\aStar$ must be nonnegative ordered.
\end{proof}

\subsection{Relaxation to QPs}
As stated before, it suffices to obtain the optimal coefficient vector for a nonnegative ordered channel vector.
Thus, in the following, we will focus on solving the problem in \eqref{equation:IntegerQP}
for a nonnegative ordered channel vector $\hBar$.
We first relax this problem to a series of QPs.

Denote the optimal coefficient vector for $\hBar$ as $\aBarStar$.
According to Theorem~\ref{theorem:aNonnegativeOrdered}, the maximum element in $\aBarStar$ is $\aBarStar(L)$.
Suppose $\aBarStar(L)$ is known to be $\bar{a}^\star_L \in \Zbb \backslash \{0\}$,
then the problem in \eqref{equation:IntegerQP} can be relaxed as a QP,
\begin{equation}
\begin{aligned}
\label{equation:RelaxedQP}
&\underset{\a}{\text{minimize}}
&&\a^T \bsG \a\\
&\text{subject to}
&&\a \in \Rbb^L\\
&
&&\a(L) = \bar{a}_L^\star
\end{aligned}
\end{equation}
where $\bsG$ is as defined in \eqref{equation:G}.
The problem is convex since $\bsG$ is positive definite.
Denote the solution of this relaxed problem as $\aBarDagger \in \Rbb^L$.
The intuition behind this relaxation is that appropriate quantization of the real-valued optimal $\aBarDagger$
with the constraint $\a(L) = \bar{a}_L^\star$ will lead to the integer-valued optimal $\aBarStar$
or at least a close-to-optimal one with a high probability.

However, since $\aBarStar(L)$ is unknown, we alternatively approximate the problem
in \eqref{equation:RelaxedQP} by solving a series of QPs,
i.e., solving the following QP multiple times for $k = 1,2,\cdots,K$.
\begin{equation}
\begin{aligned}
\label{equation:RelaxedQPs}
&\underset{\a}{\text{minimize}}
&&\a^T \bsG \a\\
&\text{subject to}
&&\a \in \Rbb^L\\
&
&&\a(L) = k
\end{aligned}
\end{equation}
Denote the solution to the above QP with the constraint $\a(L) = k$ as $\aBarDagger_k$.
For simplicity, we use $\{s_k\}$ to denote the set with elements being $s_k, k=1,2,\cdots,K$ in the following.
As long as $K \geq \aBarStar(L)$, the solution $\aBarDagger$ to the QP in \eqref{equation:RelaxedQP}
will be included in the set of solutions $\{ \aBarDagger_k \}$ to the QPs in \eqref{equation:RelaxedQPs}.
Fortunately, to obtain the solution set $\{ \aBarDagger_k \}$, it is sufficient to solve merely one QP
in \eqref{equation:RelaxedQPs} with $k = 1$, according to the following theorem.
\begin{theorem}
\label{theorem:aLinear}
Denote the solution to the QP in \eqref{equation:RelaxedQPs} with the constraint $\a(L) = k$ as $\aBarDagger_k$, then
%\begin{equation*}
%\begin{aligned}
$\aBarDagger_k = k \aBarDagger_1$.
%\end{aligned}
%\end{equation*}
\end{theorem}

\begin{proof}
\begin{align*}
\aBarDagger_k
&= \arg \min_{ \a \in \Rbb^L, \a(L) = k }
\a^T \bsG \a\\
&= k \left( \arg \min_{ \mathbf{t} \in \Rbb^L, \mathbf{t}(L) = 1 }
(k \mathbf{t}^T) \bsG (k \mathbf{t}) \right) \quad ({\rm Let\ } \a = k \mathbf{t}.)\\
&= k \left( \arg \min_{ \mathbf{t} \in \Rbb^L, \mathbf{t}(L) = 1 }
k^2 \mathbf{t}^T \bsG \mathbf{t} \right)\\
&= k \left( \arg \min_{ \mathbf{t} \in \Rbb^L, \mathbf{t}(L) = 1 }
\mathbf{t}^T \bsG \mathbf{t} \right)\\
&= k \aBarDagger_1. \qedhere
\end{align*}
\end{proof}

The closed-form expression of $\aBarDagger_1$ can be readily obtained by solving a linear system
as stated in the following theorem, which is the key for the low complexity of our method.
\begin{theorem}
\label{theorem:aBarDagger1}
Let $\aBarDagger_1$ be the optimal solution to the QP in \eqref{equation:RelaxedQPs} with the constraint $\a(L) = 1$, then
\begin{equation*}
\aBarDagger_1 =
\begin{bmatrix}
\rr\\
1
\end{bmatrix},
\end{equation*}
where
\begin{equation}
\label{equation:r}
\rr = - % do not miss the "-"
\Bigl( \bsG(1\!:\!L-1,1\!:\!L-1) \Bigr)^{-1}
\bsG(1\!:\!L-1,L).
\end{equation}
\end{theorem}

\begin{proof}
The QP in \eqref{equation:RelaxedQPs} has only an equality constraint,
and thus is linear and particularly simple \cite{Luenberger2008}.
We now derive the closed-form solution with the Lagrange multiplier method.
Let the Lagrange multiplier associated with the constraint $\a(L) = 1$ be $\lambda \geq 0$,
then the Lagrangian is
$$\mathcal{L}(\a,\lambda) = \a^T \bsG \a + \lambda \left( \a(L) - 1 \right).$$
The optimal solution can be obtained by letting the derivative of the Lagrangian be zero, i.e.,
\begin{align*}
\frac{\partial}{\partial \a}\mathcal{L}(\a,\lambda)
= (\bsG + \bsG^T) \a +
\begin{bmatrix}
\0\\
\lambda
\end{bmatrix}
= 2 \bsG \a +
\begin{bmatrix}
\0\\
\lambda
\end{bmatrix}
= \0.
\end{align*}
Let $\rr = \a(1\!:\!L-1)$, $\lambda = 2 \mu$, and write $\bsG$ and $\a$ as block matrices, then
\begin{equation*}
\begin{aligned}
\left[
\arraycolsep=1pt\def\arraystretch{1.5}
\begin{array}{c|c}
\bsG(1\!:\!L-1,1\!:\!L-1)
&\bsG(1\!:\!L-1,L)\\
\hline
\bsG(L,1\!:\!L-1)
&\bsG(L,L)
\end{array}
\right]
\begin{bmatrix}
\rr\\
1
\end{bmatrix}
+
\begin{bmatrix}
\0\\
\mu
\end{bmatrix}
= \0.
\end{aligned}
\end{equation*}
In the above equation, observe that
\begin{equation*}
\begin{aligned}
\left[
\arraycolsep=1pt\def\arraystretch{1.5}
\begin{array}{c|c}
\bsG(1\!:\!L-1,1\!:\!L-1)
&\bsG(1\!:\!L-1,L)
\end{array}
\right]
\begin{bmatrix}
\rr\\
1
\end{bmatrix}
= \0,
\end{aligned}
\end{equation*}
then
\begin{equation*}
\rr = - % do not miss the "-"
\Bigl( \bsG(1\!:\!L-1,1\!:\!L-1) \Bigr)^{-1}
\bsG(1\!:\!L-1,L),
\end{equation*}
and the results follow immediately.
\end{proof}

Calculating $\aBarDagger_1$ in Theorem~\ref{theorem:aBarDagger1} has a complexity order of
$O(L^3)$ due to the matrix inversion in the expression of $\bsr$.
We note that this complexity order can be reduced to $O(L)$ by the following lemma.
\begin{lemma}
\label{lemma:r}
Equation~\eqref{equation:r} can be expressed in a simpler form as
\begin{equation*}
\bsr = \frac{\u(L)}{1-\norm{\u(1\!:\!L-1)}^2}\u(1\!:\!L-1),
\end{equation*}
where the ``normalized" channel vector $\u$ is defined as
\begin{equation}
\label{equation:u}
\u = \sqrt{\frac{P}{1+P\norm{\hBar}^2}}\hBar.
\end{equation}
\end{lemma}

\begin{proof}
Express $\bsG$ in \eqref{equation:G} and $\bsr$ in \eqref{equation:r} in terms of $\u$,
\begin{gather*}
\bsG = \I - \u\u^T,\\
\bsr = \big(\I^{L-1} - \u(1\!:\!L-1)\u^T(1\!:\!L-1)\big)^{-1}\u(L)\u(1\!:\!L-1),
\end{gather*}
where $\I^{L-1}$ denotes the identity matrix with dimension $L-1$.
Then it is easy to verify that Lemma~\ref{lemma:r} holds.
\end{proof}

With Theorems~\ref{theorem:aLinear} and \ref{theorem:aBarDagger1},
the $K$ solutions $\{ \aBarDagger_k \}$ to the $K$ QPs in \eqref{equation:RelaxedQPs} can be easily obtained.

The next step is to quantize the real-valued approximations $\{ \aBarDagger_k \}$ to integer vectors
by applying the floor or the ceiling functions to each of the elements.
One issue that still remains is how to determine the value of $K$. Intuitively, the larger $K$, the better.
Actually, it is sufficient to set $K$ as
\begin{equation}
\begin{aligned}
\label{equation:K}
K = \arg\max_{\norm{\floor{\aBarDagger_k}}^2 < 1 + P \norm{\hBar}^2} k
\end{aligned}
\end{equation}
according to the following lemma from \cite{Nazer2011}.

\begin{lemma}
\label{lemma:aRange}
For a given channel vector $\h$, the computation rate $\mathcal{R} \left( \h, \a \right)$ is zero
if the coefficient vector $\a$ satisfies
\begin{align}
\norm{\a}^2 \geq 1 + P \norm{\h}^2.
\end{align}
\end{lemma}

\begin{remark}
\label{remark:KPractical}
For high SNR (i.e., large $P$) and large dimensions of $\hBar$, $K$ in \eqref{equation:K} can be quite huge.
However, as we will show in the next section, for i.i.d. Gaussian channel entries with high SNR,
$K$ can be set to a rather small value without degrading the average computation rate.
\end{remark}

In practice, for i.i.d. Gaussian channel entries, we set an upper bound for $K$ as $K_u$,
which is determined off-line according to the simulation results,
such that the simulated average computation rate at 20dB with $K$ being $K_u$
is greater than 99\% of that with $K$ being $K_u+1$.
We set $K_u$ based on rates at 20dB since the value of $K$ influences more the rates at larger SNR,
and 20dB is the the maximum SNR considered in this paper.
Then, we set $K$ as the maximum integer that is no greater than $K_u$ while satisfies \eqref{equation:K}
at the same time, i.e.,
\begin{equation}
\begin{aligned}
\label{equation:KPractical}
K = \arg\max_{\substack{\norm{\floor{\aBarDagger_k}}^2 < 1 + P \norm{\hBar}^2\\
k\leq K_u}} k.
\end{aligned}
\end{equation}
For implementation, $K$ can be easily determined by using a bi-section search.

\subsection{Quantization}
We propose the \emph{successive quantization} algorithm shown in Algorithm \ref{agorithm:SuccessiveQuantization}
to quantize the $K$ real-valued approximations $\{ \aBarDagger_k \}$ to integer-valued vectors
$\{ \aBarDiamond_k \}$ that serve as candidates of a suboptimal coefficient vector
$\aBarDiamond$.
For convenience, define
\begin{align}
\label{equation:f}
f(\w) \triangleq \w^T \bsG \w,
\end{align}
where $\w\in\Rbb^L$, and $\bsG$ is defined in \eqref{equation:G} with $\h$ being nonnegative ordered.
Also, let $\floor{\w}_\ell$ and $\ceil{\w}_\ell$ be the vectors generated from $\w$
by applying the floor and the ceiling operations on the $\ell$\!-th element only, respectively.

\begin{algorithm}
    \DontPrintSemicolon
    \LinesNumbered
    \SetAlgoCaptionSeparator{.}
    \SetKwInOut{Input}{Input}
    \SetKwInOut{Output}{Output}
    \Input{A real-valued vector $\aBarDagger_k \in \Rbb^L$}
    \Output{A coefficient vector $\aBarDiamond_k \in \Zbb^L$ for $\hBar$}

    \BlankLine
    \For{$\ell\leftarrow 1$ \KwTo $L-1$}{
        \eIf{$f\left(\floor{\aBarDagger_k}_\ell\right) < f\left(\ceil{\aBarDagger_k}_\ell\right)$}{\label{line:FloorCondition}
            $\aBarDagger_k \leftarrow \floor{\aBarDagger_k}_\ell$\;
        }{
            $\aBarDagger_k \leftarrow \ceil{\aBarDagger_k}_\ell$\;
        }
    }

    \BlankLine
    $\aBarDiamond_k\leftarrow \aBarDagger_k$\;
    \Return{$\aBarDiamond_k$}

    \caption{Successive Quantization}
    \label{agorithm:SuccessiveQuantization}
\end{algorithm}

To simplify the inequality condition
$f\left(\floor{\aBarDagger_k}_\ell\right) < f\left(\ceil{\aBarDagger_k}_\ell\right)$
in line \ref{line:FloorCondition} of Algorithm~\ref{agorithm:SuccessiveQuantization},
we first introduce the following lemma.
\begin{lemma}
\label{lemma:FloorCondition}
For the function $f(\w) = \w^T \bsG \w$ defined in \eqref{equation:f}
where $\bsG^T = \bsG$,
the inequality condition $f\left(\floor{\w}_\ell\right) < f\left(\ceil{\w}_\ell\right)$ is equivalent to
\begin{align}
\label{equation:FloorCondition}
2 \left(\floor{\w}_\ell\right)^T \bsG(:,\ell) + \bsG(\ell,\ell) > 0.
\end{align}
\end{lemma}

\begin{proof}
$f\left(\floor{\w}_\ell\right) < f\left(\ceil{\w}_\ell\right)$ implies $\floor{\w}_\ell \neq \ceil{\w}_\ell$,
i.e., $\w(\ell)$ is not an integer.
Let $\e_\ell \in \Rbb^L$ be the vector with only one nonzero element $\e_\ell(\ell) = 1$,
then $\ceil{\w}_\ell = \floor{\w}_\ell + \e_\ell$, and
\begin{equation*}
\begin{aligned}
&f\left(\ceil{\w}_\ell\right)
= \left(\ceil{\w}_\ell\right)^T \bsG \ceil{\w}_\ell\\
&= \left( \floor{\w}_\ell + \e_\ell \right)^T \bsG \left( \floor{\w}_\ell + \e_\ell \right)\\
&= \left(\floor{\w}_\ell\right)^T \bsG \floor{\w}_\ell
+ \left(\floor{\w}_\ell\right)^T \bsG \e_\ell
+ \e_\ell^T \bsG \floor{\w}_\ell
+ \e_\ell^T \bsG \e_\ell\\
&= f\left(\floor{\w}_\ell\right)
+ \left(\floor{\w}_\ell\right)^T \bsG(:,\ell)
+ \bsG(\ell,:) \floor{\w}_\ell
+ \bsG(\ell,\ell)\\
&= f\left(\floor{\w}_\ell\right) + 2 \left(\floor{\w}_\ell\right)^T \bsG(:,\ell) + \bsG(\ell,\ell).
\end{aligned}
\end{equation*}
Obviously, $f\left(\floor{\w}_\ell\right) < f\left(\ceil{\w}_\ell\right)$ is equivalent to
\begin{align*}
2 \left(\floor{\w}_\ell\right)^T \bsG(:,\ell) + \bsG(\ell,\ell) > 0.\tag*{\qedhere}
\end{align*}
\end{proof}

\begin{lemma}
\label{lemma:FloorCondition2}
With Lemma~\ref{lemma:FloorCondition}, the inequality condition
$f\left(\floor{\aBarDagger_k}_\ell\right) < f\left(\ceil{\aBarDagger_k}_\ell\right)$
in line~\ref{line:FloorCondition} of Algorithm~\ref{agorithm:SuccessiveQuantization}
can be simplified as
\begin{align*}
2\floor{\aBarDagger_k(\ell)} - 2\left(\left(\floor{\aBarDagger_k}_\ell\right)^T\u\right)\u(\ell) + 1 - \u(\ell)^2 < 0,
\end{align*}
where $\u$ is the normalized channel vector as defined in \eqref{equation:u}.
\end{lemma}

\begin{proof}
The proof is straightforward by writing $\bsG$ in terms of $\u$, and thus is omitted here.
\end{proof}

After the quantization, a suboptimal coefficient vector $\aBarDiamond$ for $\hBar$ is obtained with
\begin{equation}
\begin{aligned}
\label{equation:aBarDiamond}
\aBarDiamond = \arg \min_{ \a \in \{ \aBarDiamond_k \} } \a^T \bsG \a.
\end{aligned}
\end{equation}

Finally, a suboptimal coefficient vector $\aDiamond$
for the original channel vector $\h$ is recovered from
$\aBarDiamond$ according to Remark~\ref{remark:Transformation}.

We summarize our proposed \emph{QP relaxation} method in Algorithm~\ref{agrm:QPRApproachOutline}.
The pseudocode is shown in Algorithm~\ref{algorithm:QPRApproachCode},
where the function $[\bar{\w},\p]={\rm sort}(\w)$ sorts the elements in $\w$ in ascending order,
returns the sorted vector $\bar{\w}$, and stores the original indices of the elements as vector $\p$,
the function ${\rm floor}(\w)$ applies the floor operation to each element of $\w$
and returns the resulted integer vector.

\begin{algorithm}
    \DontPrintSemicolon
    \SetAlgoCaptionSeparator{.}
%    \SetAlCapFnt{\small}
%    \SetAlCapNameFnt{\small}
    \SetKwInOut{Input}{Input}
    \SetKwInOut{Output}{Output}
    \Input{A channel vector $\h \in \Rbb^L$, power $P$,\\
        an upper bound $K_u$ (determined off-line) for $K$}
    \Output{A coefficient vector $\aDiamond \in \Zbb^L$ for $\h$}

    \begin{enumerate}[leftmargin=*, rightmargin=10pt]
    \item \label{item:outline:Preprocess} Preprocess $\h$ to the nonnegative ordered $\hBar$
        with Remark~\ref{remark:Transformation}.
    \item \label{item:outline:CalculateaBarDagger1} Calculate $\aBarDagger_1$,
        with Theorem~\ref{theorem:aBarDagger1} and Lemma~\ref{lemma:r}.
    \item \label{item:outline:DetermineK} Determine $K$ with \eqref{equation:KPractical}.
    \item \label{item:outline:CalculateaBarDaggerk} Calculate $\{ \aBarDagger_k \}$,
        i.e., the real-valued approximations of the optimal coefficient vector for $\hBar$,
        using Theorem~\ref{theorem:aLinear}.
    \item \label{item:outline:Quantize} Quantize $\{ \aBarDagger_k \}$
        to integer-valued vectors $\{ \aBarDiamond_k \}$ with Algorithm~\ref{agorithm:SuccessiveQuantization}.
    \item \label{item:outline:Select} Select a vector from $\{ \aBarDiamond_k \}$ to be a suboptimal coefficient vector
        $\aBarDiamond$ for $\hBar$ using \eqref{equation:aBarDiamond}.
    \item \label{item:outline:Recover} Recover a suboptimal coefficient vector $\aDiamond$ for $\h$
        from $\aBarDiamond$ according to Remark~\ref{remark:Transformation}.
    \end{enumerate}

    \caption{The Proposed QP Relaxation Method -- Outline}
    \label{agrm:QPRApproachOutline}
\end{algorithm}

\begin{algorithm}
    \setstretch{0.8}
    \DontPrintSemicolon
    \LinesNumbered
    \SetAlgoCaptionSeparator{.}
%    \SetAlCapFnt{\small}
%    \SetAlCapNameFnt{\footnotesize}
%    \footnotesize
    \small
    \SetKwInOut{Input}{Input}
    \SetKwInOut{Output}{Output}
    \SetKwFunction{KwSign}{sign}
    \SetKwFunction{KwAbs}{abs}
    \SetKwFunction{KwSort}{sort}
    \SetKwFunction{KwFloor}{floor}
    \Input{A channel vector $\h \in \Rbb^L$, power $P$,\\
        an upper bound $K_u$ (determined off-line) for $K$}
    \Output{A coefficient vector $\aDiamond \in \Zbb^L$ for $\h$}

    \BlankLine
    \tcp{Preprocessing}
    $\bst\leftarrow$\KwSign{$\h$}
    \tcp*[h]{Get signs of entries in $\h$.}\;
    \tcp{Sort \KwAbs{$\h$} in ascending order as $\hBar$.}
    \tcp{$\p$ stores original indices.}
    $(\hBar,\p)\leftarrow$\KwSort{\KwAbs{$\h$}}\;
    $b\leftarrow 1 + P\norm{\h}^2$\tcp*[h]{A constant for efficiency}\;

    \BlankLine
    \tcp{Calculate $\aBarDagger_1$.}
    $\u\leftarrow (P/b)^{1/2}\hBar$
    \tcp*[h]{Normalized channel vector}\;
    $\bsr\leftarrow \frac{\u(L)}{1-\norm{\u(1:L-1)}^2}\u(1\!:\!L-1)$\;
    $\aBarDagger_1(1\!:\!L-1)\leftarrow \bsr$\;
    $\aBarDagger_1(L)\leftarrow 1$\;

    \BlankLine
    \tcp{Determine $K$.}
    \eIf{$\norm{\KwFloor{$K_u \aBarDagger_1$}}^2 < b$}{
        $K\leftarrow K_u$
    }(\tcp*[h]{Bi-section search}){
        $K_l\leftarrow 1$\;
        \While{$K_u\neq K_l + 1$}{
            $K\leftarrow \KwFloor{$(K_u + K_l)/2$}$\;
            \eIf{$\norm{\KwFloor{$K\aBarDagger_1$}}^2 < b$}{
                $K_l\leftarrow K$\;
            }{
                $K_u\leftarrow K$\;
            }
        }
        $K\leftarrow K_l$\;
    }

    \BlankLine
    \tcp{Quantization}
    $\aBarDiamond(1\!:\!L-1)\leftarrow \0$
    \tcp*[h]{Initialize $\aBarDiamond$ and $f_{\min}$.}\;
    $\aBarDiamond(L)\leftarrow 1$\;
    $f_{\min}\leftarrow \norm{\aBarDiamond}^2 - \big((\aBarDiamond)^T\u\big)^2$
    \tcp*[h]{$f\triangleq \a^T\bsG\a$}\;
    \For{$k\leftarrow 1$ \KwTo $K$}{
        $\aBarDagger\leftarrow k\aBarDagger_1$
        \tcp*[h]{Calculate $\aBarDagger_k$.}\;
        $d\leftarrow (\aBarDagger)^T\u$
        \tcp*[h]{Temporary variable}\;
        \For{$\ell\leftarrow 1$ \KwTo $L-1$}{
            $v\leftarrow \aBarDagger(\ell)$\;
            $\aBarDagger(\ell)\leftarrow \KwFloor{$\aBarDagger(\ell)$}$\;
            $d\leftarrow d + (\aBarDagger(\ell) - v)\u(\ell)$\;
            \If{$2\aBarDagger(\ell) - 2d\u(\ell) + 1 - \u(\ell)^2 < 0$}{
                $\aBarDagger(\ell)\leftarrow \aBarDagger(\ell) + 1$\;
                $d\leftarrow d + \u(\ell)$\;
            }
        }

        \BlankLine
        \tcp{Update the record.}
        $f\leftarrow \norm{\aBarDagger}^2 - d^2$\;
        \If{$f < f_{\min}$}{
            $\aBarDiamond\leftarrow \aBarDagger$\;
            $f_{\min}\leftarrow f$\;
        }
    }
    
    \BlankLine
    \tcp{Computation rate is $1/2\log(1/f_{\min})$.}
    \tcp{Recover the coefficient vector.}
    \For{$\ell\leftarrow 1$ \KwTo $L$}{
        $\aDiamond(\p(\ell))\leftarrow \bst(\p(\ell))\aBarDiamond(\ell)$\;
    }

    \BlankLine
    \Return{$\aDiamond$}

    \caption{The Proposed QP Relaxation Method -- Pseudocode}
    \label{algorithm:QPRApproachCode}
\end{algorithm}

\subsection{Complexity Analysis}
Here we analyze the complexity of our algorithm, in terms of the number of flops required.
Referring to the outline in Algorithm~\ref{agrm:QPRApproachOutline},
the processing of $\h$ in step~\ref{item:outline:Preprocess} involves recording the signs of the elements
and sorting the elements, and takes $O(L\log(L))$ flops.
Calculating $\aBarDagger_1$ in step~\ref{item:outline:CalculateaBarDagger1}
has a complexity of $O(L)$.
For the bi-section search applied to determine $K$ in step~\ref{item:outline:DetermineK},
the maximum number of loops required to execute is $O(\log(K_u))$, the number of flops in each loop is $O(L)$,
and thus the maximum cost is $O(\log(K_u)L)$.
Step~\ref{item:outline:CalculateaBarDaggerk} takes $O(KL)$ flops.
By introducing appropriate temporary variables $b$ and $d$
as shown in Algorithm~\ref{algorithm:QPRApproachCode},
the successive quantization of a real-valued approximation $\aBarDagger_k$ can be implemented
in an efficient way in $O(L)$ flops.
Thus, the complexity of quantizing all the $K$ real-valued approximations is $O(KL)$.
Selecting a coefficient vector from the quantized vector set in step~\ref{item:outline:Select}
has a cost of $O(KL)$.
Step~\ref{item:outline:Recover} takes $O(L)$ flops.
In summary, the complexity of the method is $O(L(\log(L)+\log(K_u)+K))$.

However, the above analyzed complexity expression involves the experiment-based $K_u$,
and its exact order with respect to the dimension $L$ is intractable.
As an alternative, we use an upper bound to approximate the cost.
According to \eqref{equation:K}, it is easy to see that $K$ and $K_u$ are at most of order $O(\sqrt{P\norm{\h}^2})$.
Then, the complexity of our method is $O(L(\log(L)+\sqrt{P\norm{\h}^2}))$.
We reserve the power $P$ in the expression
since we may also care about how the complexity varies when the SNR gets large.

In the complexity expression above, since the square root function is strictly concave,
it follows from Jensen's inequality that $\mathbb{E}(\sqrt{\norm{\h}^2})\leq\sqrt{\mathbb{E}(\norm{\h}^2)}$.
Specifically, for i.i.d. standard Gaussian channel entries, the expectation of $\norm{\h}^2$ is $L$,
and thus the corresponding average complexity of the proposed method becomes $O(L\log(L)+P^{0.5}L^{1.5})$.
It is easy to see that the complexity is of order 1.5 with respect to the dimension~$L$.

\subsection{Extension to the Complex-Valued Channel Model}
We now consider the complex-valued channel model of the AWGN networks,
and demonstrate how to apply the proposed QP relaxation method for complex-valued channels.
The complex-valued channel model is defined as below.

\begin{definition}
\label{definition:ComplexChannelModel}
\emph{(Complex-Valued Channel Model)}
In an AWGN network, each relay (indexed by $m=1,2,\cdots,M$) observes a noisy linear combination of
the transmitted signals through the channel,
\begin{align}
\label{equation:ComplexChannelModel}
\y_m = \sum_{\ell=1}^L \h_m(\ell)\x_\ell + \z_m,
\end{align}
where $\x_\ell\in\Cbb^n$ with the power constraint $\frac{1}{n}\norm{\x_\ell}^2 \leq P$
is the transmitted codeword from source $\ell$ ($\ell = 1,2,\cdots,L$),
$\h_m\in\Cbb^L$ is the channel vector to relay $m$,
$\z_m\in\Cbb^n$ is the noise vector with entries being i.i.d. Gaussian,
i.e., $\z_m\sim\mathcal{CN}\!\left(\0,\I\right)$, and $\y_m$ is the signal received at relay~$m$.
\end{definition}

Similar to what we have done for the real-valued channel model, we will focus on one relay,
and ignore the subscript ``$m$'' for notational convenience.

Writing the summation in \eqref{equation:ComplexChannelModel} in the vector product form,
\eqref{equation:ComplexChannelModel} becomes
\begin{align}
\label{equation:ComplexChannelModelVectorForm}
\y = [\x_1,\x_2,\cdots,\x_L]\h + \z.
\end{align}
It is well-known that a complex-valued channel model can be written in its real-valued equivalent form.
Let $\Re(\w)$ denote the vector composed of the real part of $\w$,
and $\Im(\w)$ denote the vector composed of the imaginary part of $\w$.
The complex-valued equation~\eqref{equation:ComplexChannelModelVectorForm} has the following
real-equivalent form
\begin{equation}
\small
\begin{aligned}
\label{equation:ComplexChannelModelRealForm}
[\Re(\y),\Im(\y)]=[\Re(\x_1),\Re(\x_2),\cdots,\Re(\x_L),\Im(\x_1),\Im(\x_2),\cdots,\Im(\x_L)]
\times
\begin{bmatrix}
\Re(\h),\Im(\h)\\
-\Im(\h),\Re(\h)
\end{bmatrix}
+ [\Re(\z),\Im(\z)].
\end{aligned}
\end{equation}
It is obvious that
\begin{equation}
\small
\begin{aligned}
\label{equation:ComplexChannelModelRealForm2}
\Re(\y)&=[\Re(\x_1),\Re(\x_2),\cdots,\Re(\x_L), \Im(\x_1),\Im(\x_2),\cdots,\Im(\x_L)]
\times
\begin{bmatrix}
\Re(\h)\\
-\Im(\h)
\end{bmatrix}
+ \Re(\z),\\
\Im(\y)&=[\Re(\x_1),\Re(\x_2),\cdots,\Re(\x_L), \Im(\x_1),\Im(\x_2),\cdots,\Im(\x_L)]
\times
\begin{bmatrix}
\Im(\h)\\
\Re(\h)
\end{bmatrix}
+ \Im(\z).
\end{aligned}
\end{equation}
Then we can view
$\begin{bmatrix}
\Re(\h)\\
-\Im(\h)
\end{bmatrix}$
and
$\begin{bmatrix}
\Im(\h)\\
\Re(\h)
\end{bmatrix}$
as two $2L$-dimensional real-valued channels,
and view $\Re(\x_\ell)$ and $\Im(\x_\ell)$ as two independent $n$-dimensional real-valued transmitted codewords.
Assume equal power allocation on the real part and the imaginary part of each transmitted codeword,
i.e., $\norm{\Re(\x_\ell)}^2=\norm{\Im(\x_\ell)}^2$,
then the power constraint of each real-valued transmitted codeword is
$\frac{1}{n}\norm{\Re(\x_\ell)}^2\leq\frac{1}{2}P$ and $\frac{1}{n}\norm{\Im(\x_\ell)}^2\leq\frac{1}{2}P$.

Based on the above interpretation, we can apply the proposed QP relaxation method to
each of the two $2L$-dimensional channels to find the corresponding coefficient vectors.
It should be noted that we only need to find the coefficient vector for one of the $2L$-dimensional channels,
which saves half of the computation cost.
Let $\a$ be a Gaussian integer, and assume the found coefficient vector for
$\begin{bmatrix}
\Re(\h)\\
-\Im(\h)
\end{bmatrix}$
is
$\begin{bmatrix}
\Re(\a)\\
-\Im(\a)
\end{bmatrix}$.
Then, according to Theorem~\ref{theorem:ProblemTransformation}, the coefficient vector for
$\begin{bmatrix}
\Im(\h)\\
\Re(\h)
\end{bmatrix}$
is
$\begin{bmatrix}
\Im(\a)\\
\Re(\a)
\end{bmatrix}$.
In this sense, for each complex-valued channel vector $\h$,
we can find a Gaussian integer as the best coefficient vector $\a$ using the QP relaxation method.

\section{Numerical Results}
\label{section:NumericalResults}
\begin{figure}[!t]
    \centering
    \includegraphics[width=240pt]{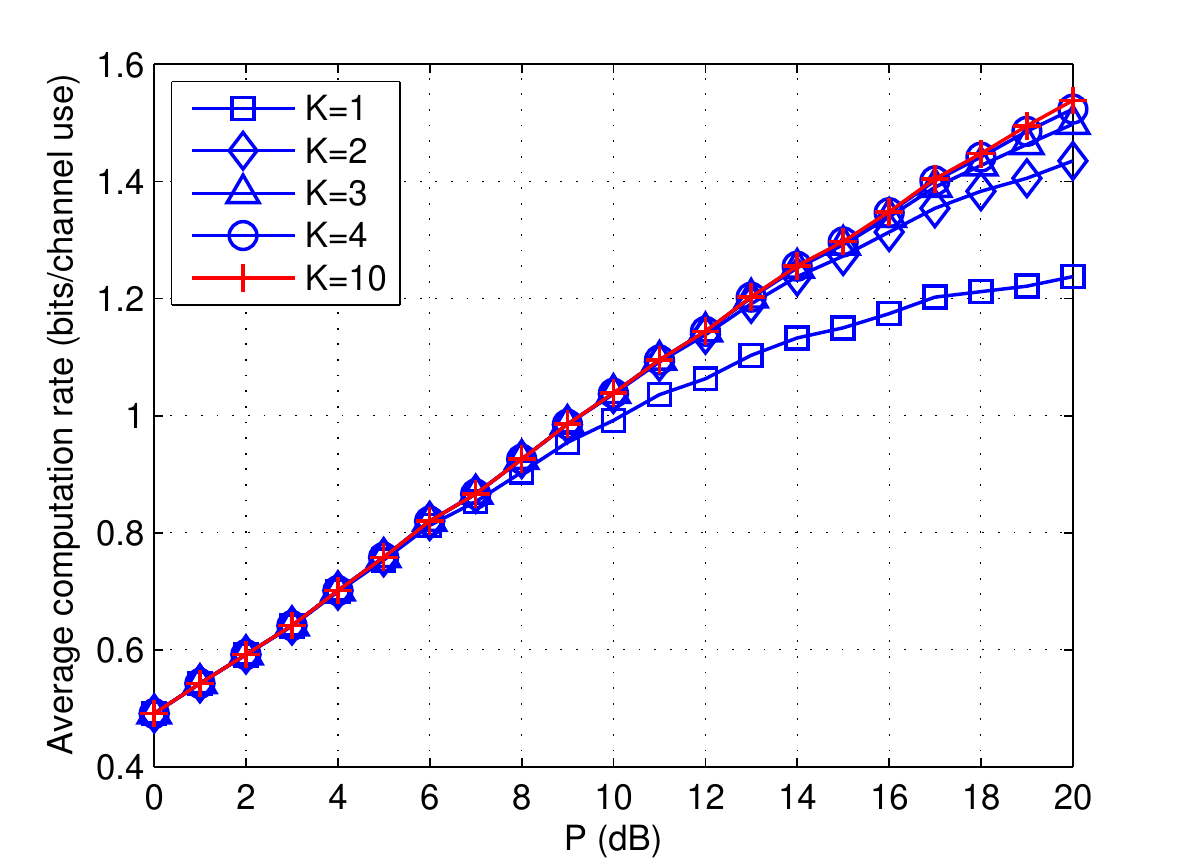}\\
    \caption{Average computation rate for $L=4$ using our QPR method with different $K$.}
    \label{figure:K_L4}
\end{figure}

In this section, we present some numerical results to demonstrate the effectiveness and efficiency of
our QP relaxation approach.
As explained before, finding the coefficient vector for a complex-valued channel can be transformed to
finding the coefficient vector for a real-valued channel.
Thus, we focus on the real-valued channels here.
We consider the case where the entries of the channel vector $\h$ are i.i.d. standard Gaussian,
i.e., $\h\sim\mathcal{N}(\0,\I)$.
In our simulations, the dimension $L$ ranges from 2 to 16, and the power $P$ ranges from 0dB to 20dB.
For a given dimension and a given power, we randomly generate $10000$ instances of the channel vector,
and apply the QP relaxation method to find the coefficient vectors,
and calculate the corresponding average computation rate.

We first show that as stated in Remark~\ref{remark:KPractical}, for high dimension and large power,
the number of real-valued approximations $K$ can be set to a rather small value
without degrading the rate apparently.
As shown in Figure~\ref{figure:K_L4}, for dimension $L=4$ and power $P$ from 0dB to 20dB,
the average computation rate quickly converges as $K$ increases from 1 to 4.
Further increasing $K$ up to 10 incurs additional computational cost
with little improvement in the average computation rate.

With the above observation, it is reasonable to introduce the upper bound $K_u$ for $K$,
and adopt the criterion in \eqref{equation:KPractical} to determine $K$.
$K_u$ can be calculated off-line by simulations prior to applying the method,
which incurs no additional processing complexity in real-time.
The values of $K_u$ according to the simulation results are listed in Table~\ref{table:Ku}.

\setlength{\tabcolsep}{3pt} % adjust the table column separation
\begin{table}[!ht]
    \caption{$K_u$ in \eqref{equation:KPractical} for the proposed QPR method}
    \label{table:Ku}
    \centering
    \begin{tabular}{c|*{15}{r}}
        \hline
        $L$ &2 & 3 & 4 & 5 & 6 & 7 & 8 & 9 & 10 & 11 & 12 & 13 & 14 & 15 & 16\\
        \hline
        $K_u$ &2 & 3 & 4 & 5 & 5 & 5 & 6 & 6 & 6 & 6 & 7 & 6 & 6 & 6 & 4\\
        \hline
    \end{tabular}
\end{table}

\begin{figure*}[!b]
    \centering
    \subfloat[$L=2$]{
        \label{figure:Rate_L2}
        \includegraphics[width=240pt]{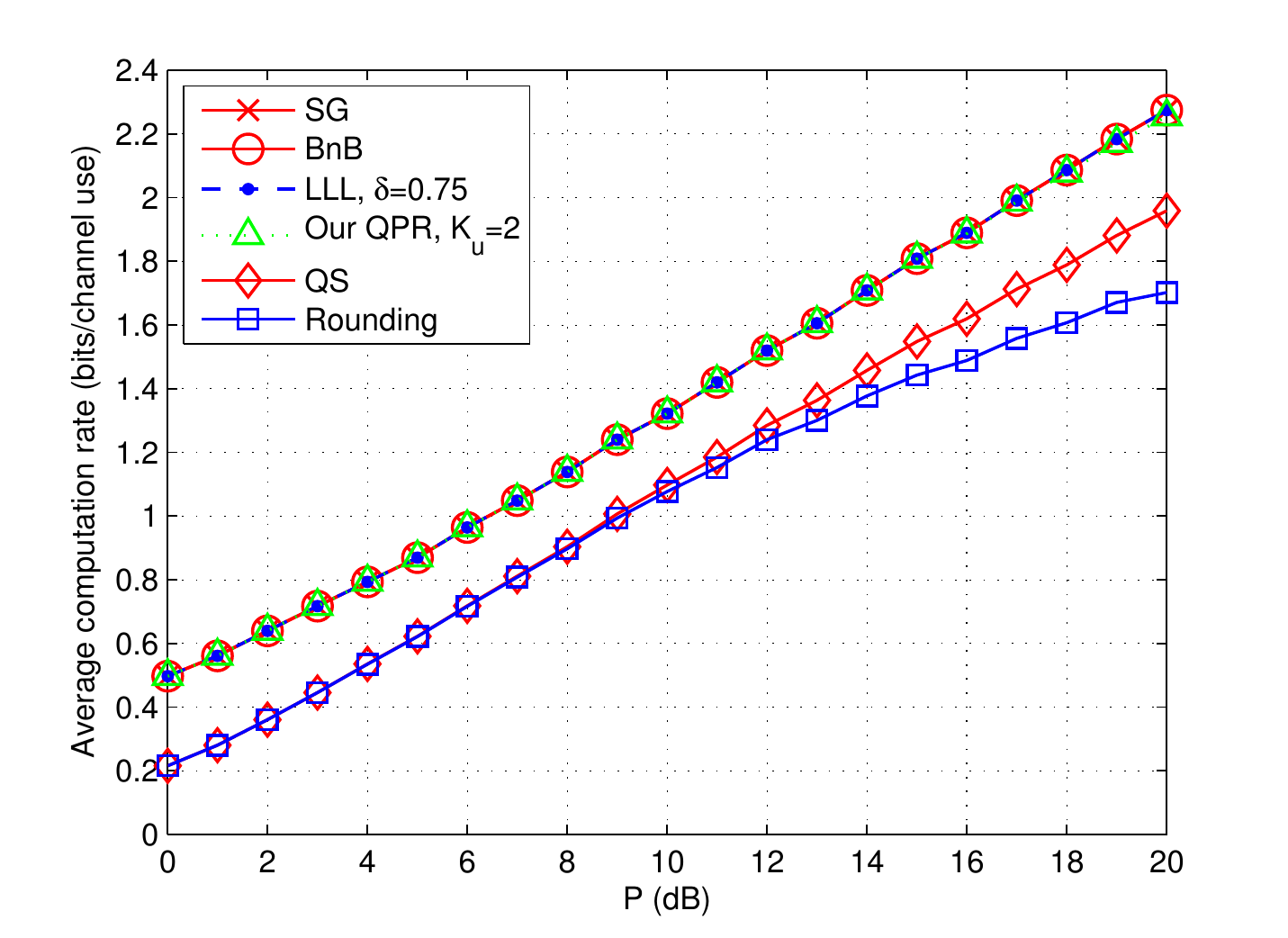}}
    \subfloat[$L=4$]{
        \label{figure:Rate_L4}
        \includegraphics[width=240pt]{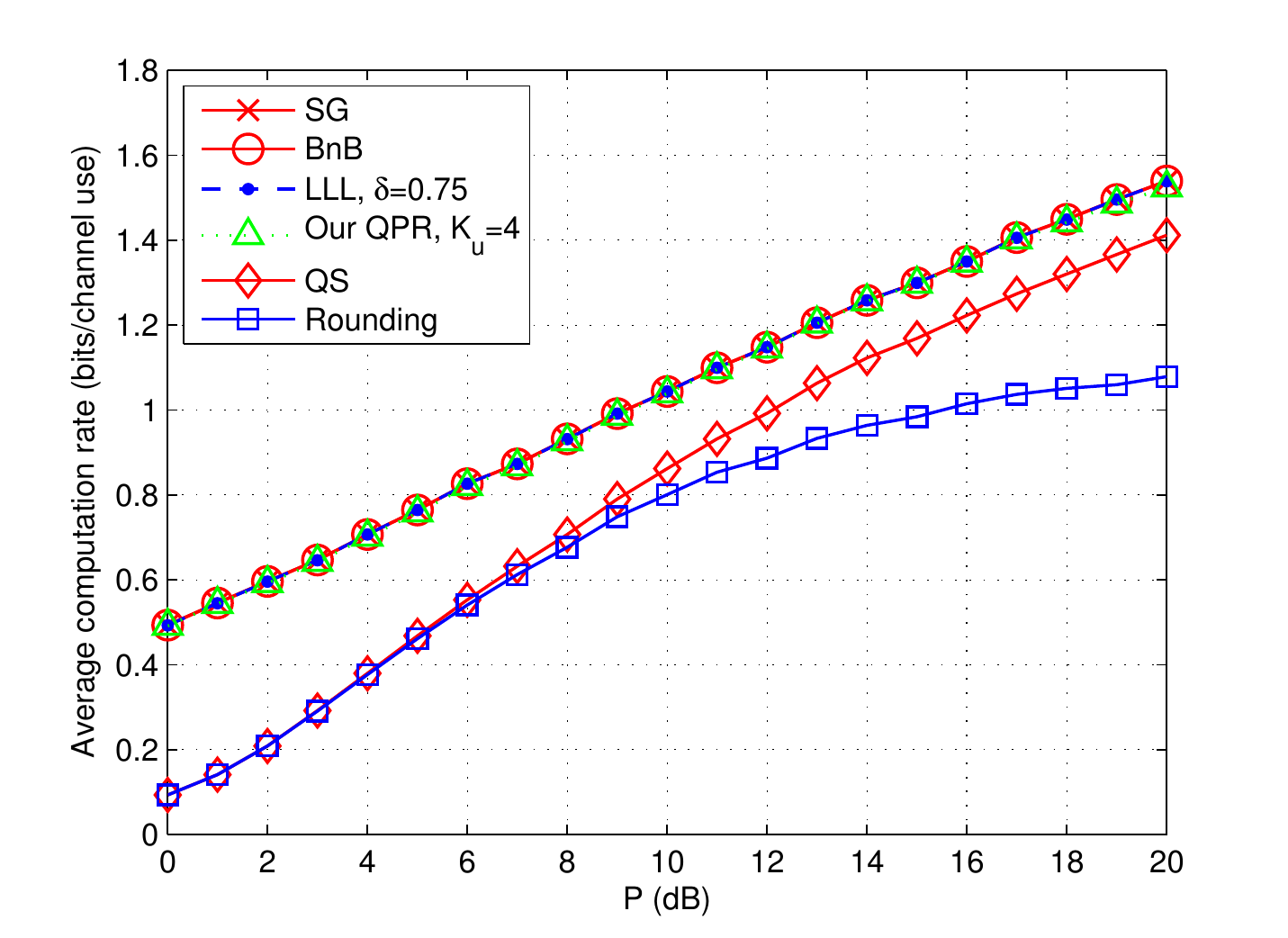}}\\
    \subfloat[$L=8$]{
        \label{figure:Rate_L8}
        \includegraphics[width=240pt]{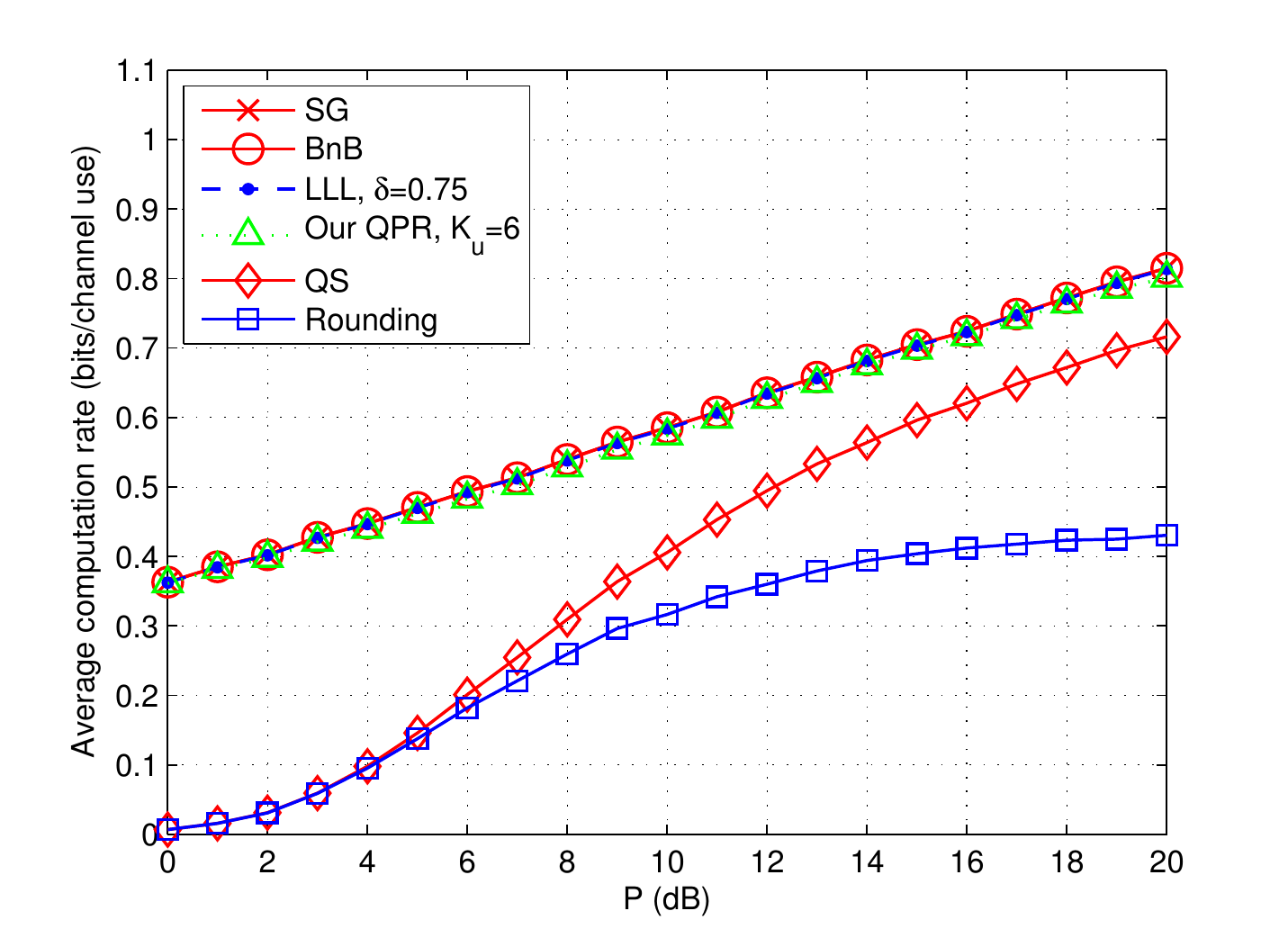}}
    \subfloat[$L=16$]{
        \label{figure:Rate_L16}
        \includegraphics[width=240pt]{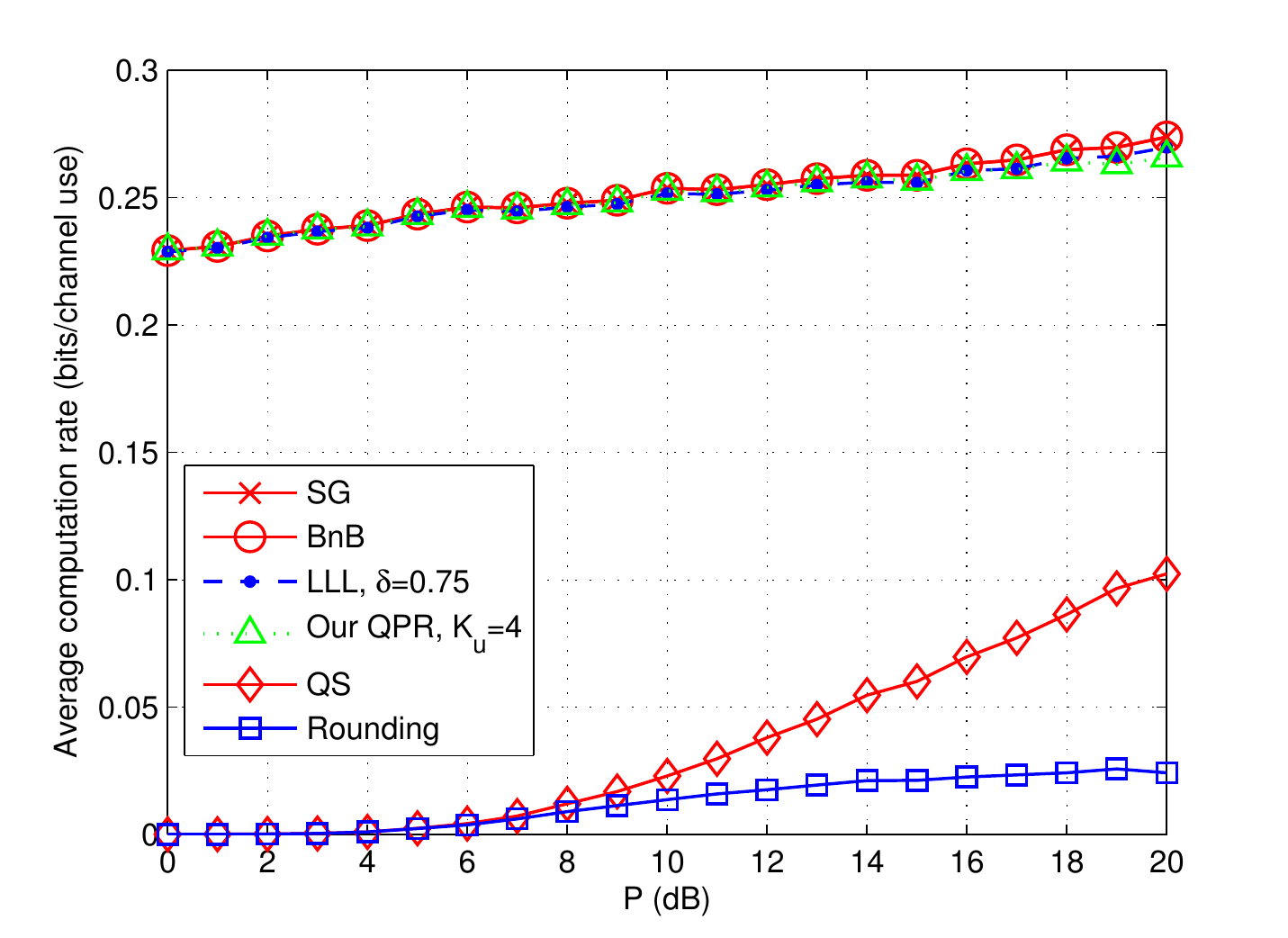}}
    \caption{Average computation rate using different methods.}
    \label{figure:Rate}
\end{figure*}

We then show the effectiveness of our method by comparing the average computation rate with
those of other existing methods.
The methods covered include the following.
\begin{itemize}
\item Our QP relaxation (QPR) method that gives the suboptimal solution.
\item The branch-and-bound (BnB) method proposed by Richter {\it et al.} in \cite{Richter2012}
that provides the optimal solution.
\item The method developed by Sahraei and Gastpar in \cite{Sahraei2014}
that finds the optimal solution with an average-case complexity of $O(P^{0.5}L^{2.5})$ for i.i.d. Gaussian channel entries.
We refer to this method as the ``SG" method for short.
\item The LLL method proposed by Sakzad {\it et al.} in \cite{Sakzad2012},
which is based on the LLL lattice reduction (LR) algorithm.
The parameter $\delta$ in the LLL LR algorithm is set as 0.75 since further increasing $\delta$
towards 1 achieves little gain in the computation rate but requires more computation labor.
Although the LLL LR algorithm has known average complexity for some cases~\cite{Daude1994, Ling2007},
its average complexity for our case is unknown, and the worst-case complexity could be unbounded~\cite{Jalden2008}.
\item The quantized search (QS) method developed by Sakzad {\it et al.} in \cite{Sakzad2012}.
The search consists of two phases: 1) an integer $\alpha_0$ between 1 and $\floor{P^{1/2}}$
that provides the maximum rate is selected as the initial value of the amplifying factor $\alpha$;
2) the amplifying factor is then refined by searching in $[\alpha_0-1,\alpha_0+1]$ with a step size 0.1.
After the amplifying factor $\alpha$ is determined, the coefficient vector $\a$ is set as $\round{\alpha\h}$.
An improved version of the QS method is the quantized exhaustive search (QES) method
proposed in \cite{Sakzad2014}, which was developed for complex-valued channels.
\item The rounding method that simply sets the coefficient vector by rounding the channel vector
to an integer-valued vector.
\end{itemize}

As shown in Figure~\ref{figure:Rate}, the optimal methods, i.e., the BnB method and the SG method,
always provide the highest average computation rates for all dimensions and over the whole SNR regime,
as expected.
The LLL method provides close-to-optimal average computation rates.
Our proposed QPR method also offers close-to-optimal average computation rates
for almost all the dimensions and SNR values considered,
except that its performance degrades a little bit for high dimensions at high SNR
as shown in Figure~\ref{figure:Rate_L16}.
The performance of our QPR method improves slightly
compared with the version we presented in \cite{Zhou2014, Wen2015}.
The reason is that here we initialize the output coefficient vector as $[0,\cdots,0,1]^T$,
which definitely results non-zero computation rate,
while in the previous version the output coefficient vector could result zero computation rate.

\begin{figure}[!p]
    \centering
    \subfloat[$P=0$ dB]{
        \label{figure:Time_PdB0}
        \includegraphics[width=240pt]{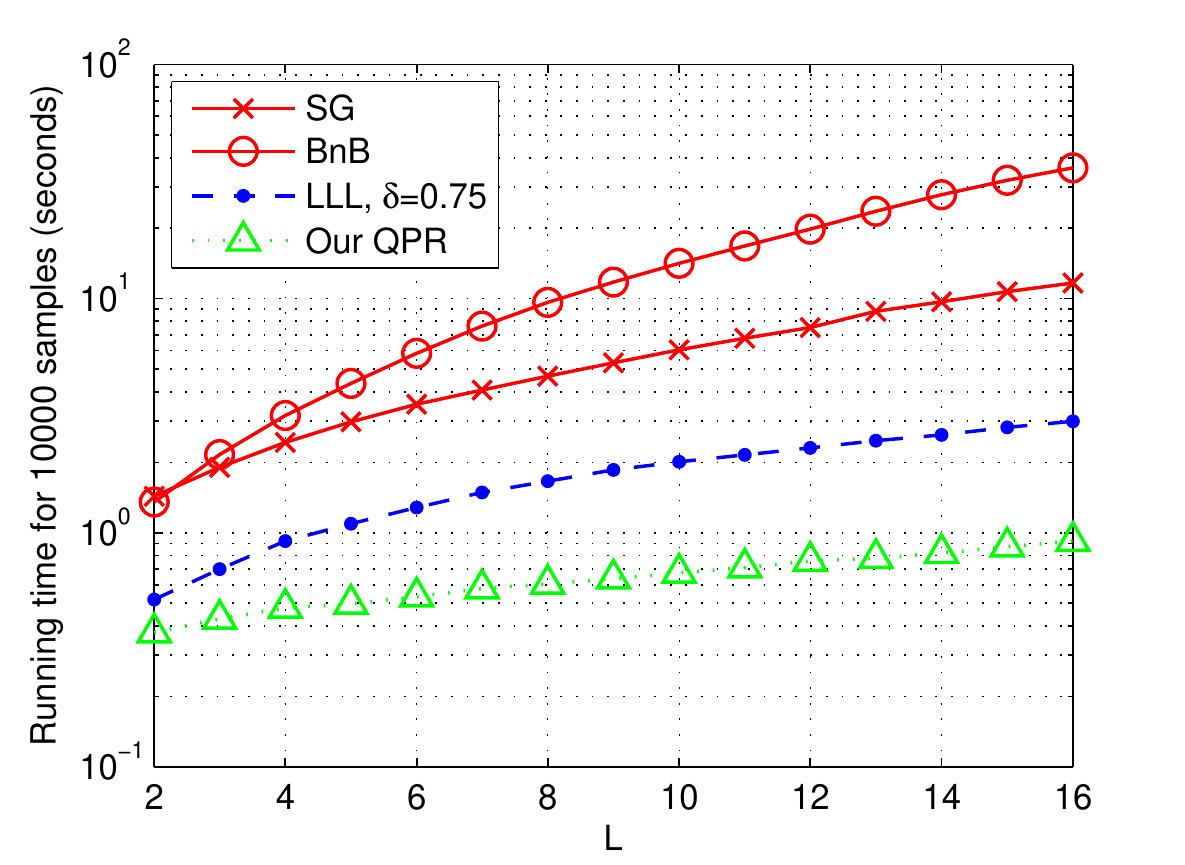}}\\
    \subfloat[$P=10$ dB]{
        \label{figure:Time_PdB10}
        \includegraphics[width=240pt]{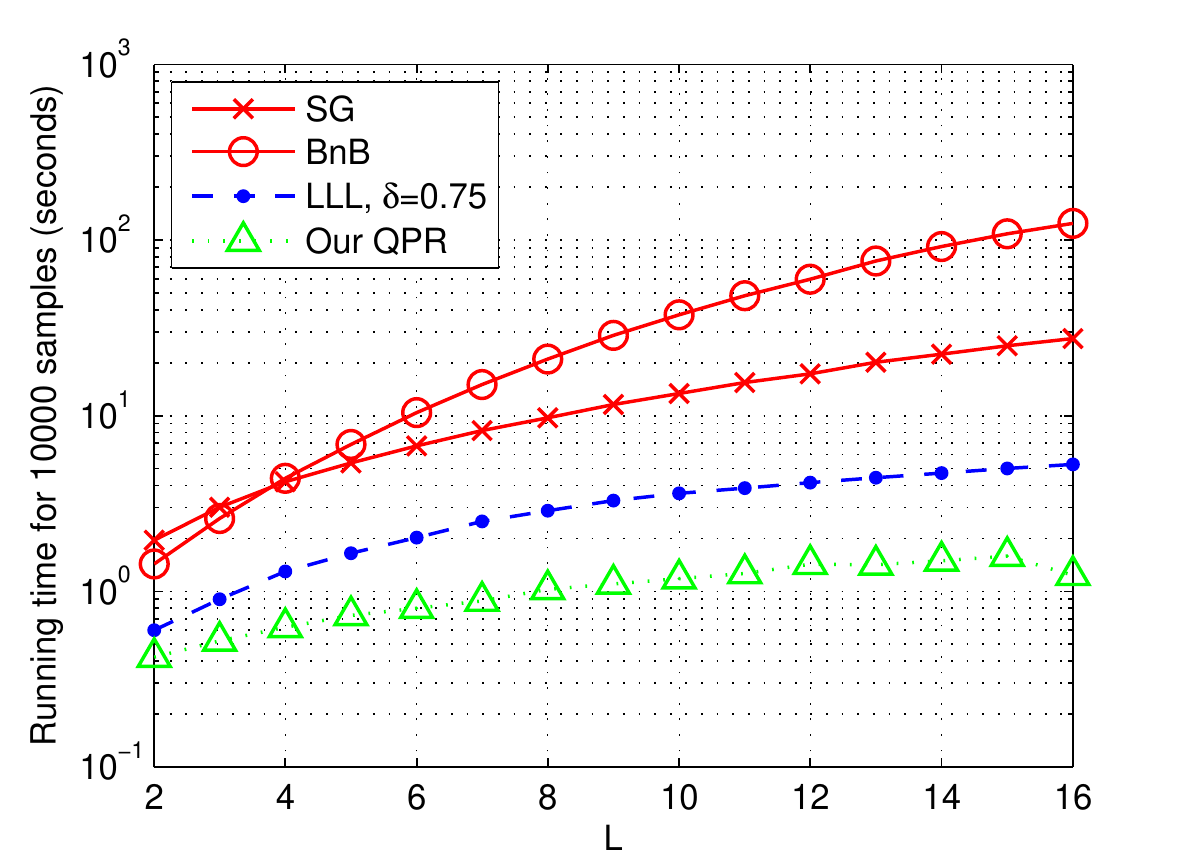}}\\
    \subfloat[$P=20$ dB]{
        \label{figure:Time_PdB20}
        \includegraphics[width=240pt]{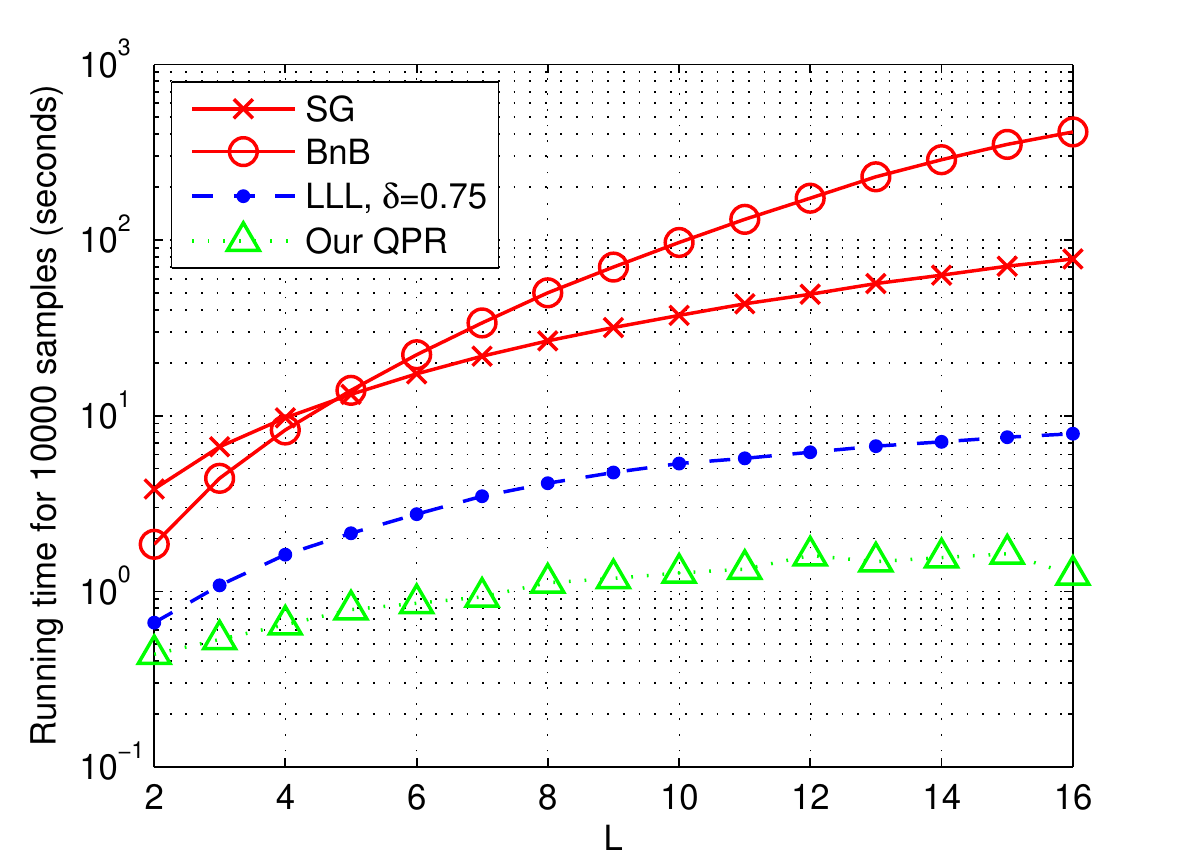}}
    \caption{Running time for $10000$ samples using different methods.}
    \label{figure:Time}
    \vspace{-10pt}
\end{figure}

Finally, we demonstrate the efficiency of the proposed QPR method by comparing
the running time of finding the coefficient vectors for 10000 channel vector samples.
The methods considered include those that provide optimal rates and close-to-optimal rates,
i.e., the SG method, the BnB method, the LLL method, and our QPR method.
The running time varies for different SNR values,
and thus we compare the running time with $P$ being 0dB, 10dB and 20dB.
As shown in Figure~\ref{figure:Time}, the proposed QPR method is much more efficient than
all the other methods, especially for high dimensions.
Specifically, the running time of the optimal methods can be one scale larger than that of the QPR method.

In summary, for i.i.d. Gaussian channel entries,
our proposed QPR method offers close-to-optimal average computation rates
with a much lower complexity than that of the existing optimal and close-to-optimal methods.

\section{Conclusions}
\label{section:Conclusions}
In this paper, we considered the compute-and-forward network coding design problem
of finding the optimal coefficient vector that maximizes the computation rate at a relay,
and developed the quadratic programming (QP) relaxation method that finds a high quality suboptimal solution.
We first revealed some useful properties of the problem, and relaxed the problem to a series of QPs.
We then derived the closed-form solutions of the QPs, which is the key to the efficiency of our method,
and proposed a successive quantization algorithm to quantize the real-valued solutions to integer vectors
that serve as candidates of the coefficient vector.
Finally, the candidate that maximizes the computation rate is selected as the best coefficient vector.
For $L$-dimensional channel vectors with i.i.d. Gaussian entries,
the average-case complexity of the proposed QP relaxation method is of order 1.5
with respect to the dimension $L$.
Numerical results demonstrated that our QP relaxation method offers close-to-optimal computation rates,
and is much more computationally efficient than the existing methods that provide the optimal computation rates
as well as the LLL method that also provides close-to-optimal computation rates.

%\bibliographystyle{IEEEtran}
%\bibliography{C:/Users/Alan/Dropbox/Research/ref/cf_qpr} % On laptop
\bibliographystyle{IEEEtran}
\bibliography{C:/Users/bzhouab/Dropbox/Research/ref/cf_qpr} % On desktop

\begin{thebibliography}{10}
\providecommand{\url}[1]{#1}
\csname url@samestyle\endcsname
\providecommand{\newblock}{\relax}
\providecommand{\bibinfo}[2]{#2}
\providecommand{\BIBentrySTDinterwordspacing}{\spaceskip=0pt\relax}
\providecommand{\BIBentryALTinterwordstretchfactor}{4}
\providecommand{\BIBentryALTinterwordspacing}{\spaceskip=\fontdimen2\font plus
\BIBentryALTinterwordstretchfactor\fontdimen3\font minus
  \fontdimen4\font\relax}
\providecommand{\BIBforeignlanguage}[2]{{%
\expandafter\ifx\csname l@#1\endcsname\relax
\typeout{** WARNING: IEEEtran.bst: No hyphenation pattern has been}%
\typeout{** loaded for the language `#1'. Using the pattern for}%
\typeout{** the default language instead.}%
\else
\language=\csname l@#1\endcsname
\fi
#2}}
\providecommand{\BIBdecl}{\relax}
\BIBdecl

\bibitem{Nazer2008}
B.~Nazer and M.~Gastpar, ``Compute-and-forward: Harnessing interference with
  structured codes,'' in \emph{2008 IEEE International Symposium on Information
  Theory (ISIT)}, 2008, pp. 772--776.

\bibitem{Nazer2011}
------, ``Compute-and-forward: Harnessing interference through structured
  codes,'' \emph{IEEE Transactions on Information Theory}, vol.~57, no.~10, pp.
  6463--6486, Oct 2011.

\bibitem{Wei2012}
L.~Wei and W.~Chen, ``Compute-and-forward network coding design over
  multi-source multi-relay channels,'' \emph{IEEE Transactions on Wireless
  Communications}, vol.~11, no.~9, pp. 3348--3357, 2012.

\bibitem{Richter2012}
J.~Richter, C.~Scheunert, and E.~Jorswieck, ``An efficient branch-and-bound
  algorithm for compute-and-forward,'' in \emph{2012 IEEE
  23\textsuperscript{rd} International Symposium on Personal Indoor and Mobile
  Radio Communications (PIMRC)}, Sep. 2012, pp. 77--82.

\bibitem{Sahraei2014}
S.~Sahraei and M.~Gastpar, ``Compute-and-forward: Finding the best equation,''
  in \emph{to appear in 52nd Annual Allerton Conference on Communication,
  Control, and Computing, Champaign, Illinois, USA}, 2014.

\bibitem{Zhang2012}
W.~Zhang, S.~Qiao, and Y.~Wei, ``{HKZ} and {M}inkowski reduction algorithms for
  lattice-reduction-aided {MIMO} detection,'' \emph{IEEE Transactions on Signal
  Processing}, vol.~60, no.~11, pp. 5963--5976, 2012.

\bibitem{Lenstra1982}
\BIBentryALTinterwordspacing
A.~K. Lenstra, H.~W. Lenstra, Jr., and L.~Lov\'asz,
  ``\BIBforeignlanguage{eng}{Factoring polynomials with rational
  coefficients.}'' \emph{\BIBforeignlanguage{eng}{Mathematische Annalen}}, vol.
  261, pp. 515--534, 1982. [Online]. Available:
  \url{http://eudml.org/doc/182903}
\BIBentrySTDinterwordspacing

\bibitem{Vetter2009}
H.~Vetter, V.~Ponnampalam, M.~Sandell, and P.~Hoeher, ``Fixed complexity {LLL}
  algorithm,'' \emph{IEEE Transactions on Signal Processing}, vol.~57, no.~4,
  pp. 1634--1637, April 2009.

\bibitem{Ling2013}
C.~Ling, W.~H. Mow, and N.~Howgrave-Graham, ``Reduced and fixed-complexity
  variants of the {LLL} algorithm for communications,'' \emph{IEEE Transactions
  on Communications}, vol.~61, no.~3, pp. 1040--1050, 2013.

\bibitem{Chang2013}
X.-W. Chang, J.~Wen, and X.~Xie, ``Effects of the {LLL} reduction on the
  success probability of the {B}abai point and on the complexity of sphere
  decoding,'' \emph{IEEE Transactions on Information Theory}, vol.~59, no.~8,
  pp. 4915--4926, Aug 2013.

\bibitem{Gan2009}
Y.~H. Gan, C.~Ling, and W.~H. Mow, ``Complex lattice reduction algorithm for
  low-complexity full-diversity mimo detection,'' \emph{IEEE Transactions on
  Signal Processing}, vol.~57, no.~7, pp. 2701--2710, 2009.

\bibitem{Sakzad2012}
A.~Sakzad, E.~Viterbo, Y.~Hong, and J.~Boutros, ``On the ergodic rate for
  compute-and-forward,'' in \emph{2012 International Symposium on Network
  Coding (NetCod)}, 2012, pp. 131--136.

\bibitem{Wen2015}
J.~Wen, B.~Zhou, W.~H. Mow, and X.-W. Chang, ``Compute-and-forward protocol
  design based on improved sphere decoding,'' in \emph{(to appear) Proceedings
  of IEEE International Conference on Communications}, 2015.

\bibitem{Zhou2014}
B.~Zhou and W.~H. Mow, ``A quadratic programming relaxation approach to
  compute-and-forward network coding design,'' in \emph{2014 IEEE International
  Symposium on Information Theory (ISIT)}, June 2014, pp. 2296--2300.

\bibitem{Luenberger2008}
D.~G. Luenberger and Y.~Ye, \emph{Linear and Nonlinear Programming},
  3rd~ed.\hskip 1em plus 0.5em minus 0.4em\relax US: Springer, 2008.

\bibitem{Daude1994}
\BIBentryALTinterwordspacing
H.~Daude and B.~Vallee, ``An upper bound on the average number of iterations of
  the {LLL} algorithm,'' \emph{Theoretical Computer Science}, vol. 123, no.~1,
  pp. 95 -- 115, 1994. [Online]. Available:
  \url{http://www.sciencedirect.com/science/article/pii/030439759490071X}
\BIBentrySTDinterwordspacing

\bibitem{Ling2007}
C.~Ling and N.~Howgrave-Graham, ``Effective {LLL} reduction for lattice
  decoding,'' in \emph{2007 IEEE International Symposium on Information Theory
  (ISIT)}, June 2007, pp. 196--200.

\bibitem{Jalden2008}
J.~Jalden, D.~Seethaler, and G.~Matz, ``Worst- and average-case complexity of
  {LLL} lattice reduction in mimo wireless systems,'' in \emph{2008 IEEE
  International Conference on Acoustics, Speech and Signal Processing
  (ICASSP)}, 2008, pp. 2685--2688.

\bibitem{Sakzad2014}
\BIBentryALTinterwordspacing
A.~Sakzad, E.~Viterbo, J.~J. Boutros, and Y.~Hong, ``Phase precoding for the
  compute-and-forward protocol,'' \emph{CoRR}, vol. abs/1404.4157, 2014.
  [Online]. Available: \url{http://arxiv.org/abs/1404.4157}
\BIBentrySTDinterwordspacing

\end{thebibliography}

\end{document}